\documentclass{llncs}
\usepackage{main}

\begin{document}

\title{An Efficient VCGen-based Modular Verification of Relational Properties
\thanks{
This version of the contribution has been accepted for publication, after peer review (when applicable) but is not the Version of Record and does not reﬂect post-acceptance improvements, or any corrections.
The Version of Record is available online at: https://doi.org/10.1007/978-3-031-19849-6_28.}
}

\author{
  \mbox{}
  \hspace{-4mm}
  Lionel Blatter\inst{1,2} \and
  Nikolai Kosmatov\inst{3,4} \and
  Virgile Prevosto\inst{3} \and\\
  Pascale Le Gall\inst{5}
}

\institute{
  Karlsruhe Institute of Technology, 76131, Karlsruhe, Germany \\
  \email{firstname.lastname@kit.edu}
  \and
  Max Planck Institute for Security and Privacy, 44799, Bochum, Germany
  \and
  Université Paris-Saclay, CEA, List, 91120, Palaiseau, France\\
  \email{firstname.lastname@cea.fr}
  \and
  Thales Research \& Technology, 91120, Palaiseau, France 
  \and
  CentraleSupélec, Université Paris-Saclay, 91190 Gif-sur-Yvette France\\
  \email{firstname.lastname@centralesupelec.fr}
}

\maketitle
\vspace{-8mm}
\begin{abstract}
  Deductive verification typically relies on function contracts that 
  specify the behavior of each function for a single function call.
  \emph{Relational properties} link several function calls together
  within a single specification. They can express more advanced properties of a
  given function, such as non-interference, continuity, or monotonicity, or relate
  calls to different functions, possibly run in parallel, for instance, to show 
  the equivalence of two implementations. 
  However, relational properties cannot be expressed and
  verified directly in the
  traditional setting of modular deductive verification.
  Recent work proposed a new technique for relational property verification 
  that relies on a verification condition generator
  to produce logical formulas that must be verified
  to ensure a given relational property.  
  This paper presents an overview of this approach and 
  proposes important enhancements.  
  We integrate an optimized verification condition generator
  and extend the underlying theory to show how relational properties can be proved 
  in a modular way, where one relational property can be used to prove another one,
  like in modular verification of function contracts. 
  Our results have been fully formalized and proved sound in the \coq proof assistant.
\end{abstract}

\vspace{-5mm}
\section{Introduction}
\vspace{-2mm}
\label{sec:intro}
Modular deductive verification~\cite{DBLP:journals/cacm/Hoare69}
is used to prove that every function $f$ of a given program
respects its \emph{contract}.
Such a contract is, basically, an implication: if the given \emph{precondition}
is true before a call to $f$ and the call
terminates\footnote{Termination can be either assumed (partial correctness)
  or proved separately (full correctness) in a classical way~\cite{Floyd1967};
  for the purpose of this paper we can assume it.},
the given \emph{postcondition} is true when $f$ returns control to the caller.
However, some kinds of properties
are not easily reducible to a single function call.
Indeed, it is often necessary to express a property
that involves several functions, possibly executed in parallel,
or relates the results of several calls to the same
function for different arguments.
Such properties are known as
\emph{relational properties}~\cite{DBLP:conf/popl/Benton04}.

Examples of such relational properties include monotonicity
({i.e.} $x\leq y\Rightarrow f(x) \leq f(y)$), involving 2 calls,
or transitivity
($\cmp(x,y) \geq 0 \wedge \cmp(y,z) \geq 0 \Rightarrow \cmp(x,z)\geq 0$),
involving 3 calls.
In secure information flow~\cite{DBLP:journals/mscs/BartheDR11},
\emph{non-interference} is also a relational property. Namely,
given a partition of program variables between high-security variables
and low-security variables,
a program is said to be non-interferent if 
any two executions starting from states in which the low-security
variables have the same initial values will end up
in a final state where the low-security variables have the same values.
In other words, high-security variables cannot interfere with low-security
ones.


\begin{figure}[tb]
\begin{tabular}{l|c}
\begin{minipage}{0.23\textwidth}
//Command $\csum$:	\\
$\begin{array}{l@{}l}
	&\wif{x_1< x_2}{\\
		&\qquad x_3:= x_3 + x_1;\\
		&\qquad x_1:= x_1 + 1;\\
		&\qquad \wcall{\ysum}\\
		&}{\ \wskip\ }\\
\end{array}$
\end{minipage}
&\,
\begin{minipage}{0.74\textwidth}
Relational property $\RelPI$ between commands $\csI$ and $\csII$:	\\
\begin{equation*}
\left\{
\begin{array}{l@{}l}
\ben{\locval_2}{1} &= \ben{\locval_2}{2}
\end{array}
\right\}
\ben{
	\begin{array}{l@{}l}
	&\mbox{// $\csI$:}\\
	&x_1:= 1;\\
	&x_3:= 0;\\
	&\wcall{\ysum}\\
	\end{array}
}{1}
\sim
\ben{
	\begin{array}{l@{}l}
	&\mbox{// $\csII$:}\\
	&x_1:= 0;\\
	&x_3:= 0;\\
	&\wcall{\ysum}\\
	\end{array}
}{2}
\left\{
\begin{array}{l@{}l}
\ben{\locval_3}{1} &= \ben{\locval_3}{2}
\end{array}
\right\}
\end{equation*}
\end{minipage}
\end{tabular}
\vspace{-2mm}
\caption{Recursive command $\csum$, associated as a body with procedure name $\ysum$, 
and relational property $\RelPI$ between two commands, 
denoted $\csI$ and $\csII$, involving a call to this procedure.}
\vspace{-5mm}
\label{fig:ex-intro}
\end{figure}

\paragraph{Motivation.}
Lack of support for relational properties in verification tools
was already faced by industrial users (e.g. in~\cite{BishopBC13} for C programs).
The usual way to deal with this limitation is to use
\emph{self-composition}~\cite{DBLP:journals/mscs/BartheDR11,DBLP:conf/fm/SchebenS14,blatterKGP17},
product programs~\cite{DBLP:conf/fm/BartheCK11} or
other self-composition variants~\cite{ShemerCAV2021}.
Those techniques are based on code transformations
that are relatively tedious and error-prone. Moreover, 
they are hardly applicable in practice to real-life programs 
with pointers like in C.
Namely, self-composition requires that the compared executions
operate on completely {separated} (i.e. disjoint) memory areas, which might
be extremely difficult to ensure for complex programs with pointers.
Modular verification of relational properties is another important 
feature: the user may want to rely on some relational properties in order
to verify some other ones.

\begin{example}[relational property]
\label{ex:rel-contract}
Figure~\ref{fig:ex-intro} shows an example of a recursive command (that is, program)
$\csum$. 
We clearly distinguish the name and the body of a procedure. 
The procedure named $\ysum$ is assumed to have command $\csum$ as its body, 
so that $\csum$ recursively calls itself. 
Given three global integer variables $\locval_1$, $\locval_2$ and $\locval_3$, command $\csum$
adds to $\locval_3$ (used as an accumulator) 
the sum $\locval_1+(\locval_1+1)+\dots+(\locval_2-1)$
if $\locval_1<\locval_2$, and has no effect otherwise.

Figure~\ref{fig:ex-intro} also shows an example of a relational property 
$\RelPI$ (inspired by~\cite{DBLP:conf/fm/BartheCK11})
stating the equivalence of two commands $\csI$ and $\csII$
(assumed to be run on separate memory states), which 
assign $\locval_1$ and $\locval_3$ before calling $\ysum$.
The relational property is written here in Benton's notation~\cite{DBLP:conf/popl/Benton04}:
tags $\ben{}{1}$ and $\ben{}{2}$ are used to 
distinguish the programs linked by the property.
When variables of the linked programs have the same names,
such a tag after a variable name also helps to distinguish
the instance of the variable used
in the relational precondition and postcondition 
(written in curly braces, resp., on the left and on the right).
Property $\RelPI$ states that if $\locval_2$ has the same value before 
the execution of $\csI$ and before the execution of $\csII$, 
then  $\locval_3$ will have the same 
value after their executions.
Indeed, 
$\csI$ will compute in 
$\locval_3$ the sum $1+2+\dots+(\locval_2-1)$, while 
$\csII$ will compute in 
$\locval_3$ the sum $0+1+2+\dots+(\locval_2-1)$.

In this paper, we will 
show how relational property $\RelPI$ 
can be verified using another relational property $\RelPIII$ 
linking two runs of $\csum$ rather than using
a full functional contract of $\csum$. More precisely,
$\RelPIII$ (that will be formally defined below in Fig.~\ref{fig:ex-contract-rp2})
generalizes the situation of $\RelPI$ and
states that the resulting value of $\locval_3$ after two runs of $\csum$
will be the same if the initial state of the second run 
is exactly one iteration of $\csum$ behind that of the first run. 
\hfil\qed
\end{example}

\vspace{-2mm}
\paragraph{Approach.}
Our recent work~\cite{BlatterKPLiFM22} proposed
an alternative to self-composition
that is not based on code transformation or relational rules.
It directly relies on a standard verification condition  generator (VCGen)
to produce logical formulas to be verified (typically, with an
automated prover) to ensure a given relational property. 
This approach requires no extra code processing 
(such as sequential composition of programs or variable renaming).
Moreover, no additional separation hypotheses
are required. 
The locations of each program are separated by construction: 
each program  
has its own memory state.
This approach has been formalized on a minimal language $\langname$,
representative of the 
main issues relevant for relational property verification.
$\langname$ is a standard \textsc{While} language extended 
with annotations, procedures and pointers.
Notably, the presence of dereferences and address-of operations makes
it representative of various aliasing problems 
with (possibly, multiple) pointer dereferences of a real-life
language like C.
An example of a relational property for programs with pointers was
given in~\cite{BlatterKPLiFM22}. 
We formalize the proposed approach and prove
its soundness in the \coq proof assistant~\cite{Coq}.
Our \coq development\footnote{
Available 
at 
\url{https://github.com/lyonel2017/Relational-Spec/}.}
contains about 3700 lines.
%

\vspace{-2mm}
\paragraph{Contributions.}
We give an overview of the VCGen-based approach for relational property verification (presented in~\cite{BlatterKPLiFM22}) and enhance the underlying theory with several new features. 
The new technical contributions of this paper include:
\begin{itemize}
\item a \coq formalization and proof of soundness 
of an optimized VCGen for
language \langname, and its extension to the verification of relational properties;
\item an extension of the framework allowing not only to \emph{prove} relational properties, but
also to \emph{use} them as hypotheses in the following proofs;
\item a \coq formalization of the extended theory.
\end{itemize}
\noindent
We also provide an illustrative example and, as another minor extension, 
add the capacity to refer to old values of variables
in postconditions. 

\vspace{-3mm}
\paragraph{Outline.}
Section~\ref{sec:background}
introduces the imperative language $\langname$ used in this work.
Functional correctness 
is defined in Section~\ref{sec:funct-corr}.
The extension of functional correctness
to relational properties is presented in Section~\ref{sec:rel-prop}. Then, we prove
the soundness of an optimized VCGen 
in Section~\ref{sec:verif-cond-gener},
and show how it can be soundly extended to verify relational properties 
in Section~\ref{sec:relat-prop-verif}. Finally, we 
present related work 
in Section~\ref{sec:related}
and
concluding remarks in Section~\ref{sec:conclusion}.


\vspace{-2mm}
\section{Syntax and Semantics of the Considered Language $\langname$}
\vspace{-2mm}
\label{sec:background}
\label{sec:considered-language}


\subsection{Locations, States, and Procedure Contracts}
\vspace{-1mm}
\label{sec:loc-state-contract}
\label{sec:program-grammar}

We denote by $\Natset=\{0,1,2,\dots\}$ the set of natural numbers,
by $\Natset^*=\{1,2,\dots\}$ the set of nonzero natural numbers,
and by $\Boolset = \{\True,\False\}$ the set of Boolean values.
Let $\Locval$ be the set of program \emph{locations} 
and
$\Loccom$ the set of \emph{program (procedure) names}, and let
$\locval,\locval',\locval_1, ...$ and $\loccom,\loccom',\loccom_1, ...$
denote  metavariables ranging over those respective sets.
We assume that there exists a bijective function 
$\Natset\to\Locval$, so that $\Locval=\{\locval_i\,|\,i\in\Natset\}$.
Intuitively, we can see $i$ as the \emph{address} of location $\locval_i$.

Let $\Memvar$ be the set of functions 
$\memvar:\Natset \to \Natset$, called \emph{memory states}, and let
$\memvar,\memvar',\memvar_1, ...$ denote metavariables ranging over $\Memvar$.
A state $\memvar$ maps a location to a value using its address:
location $\locval_i$ has value $\memvar(i).$

We define the \emph{update} operation of a memory state
$\setmem{\mem}{i}{n}$, also denoted by $\memvar[i/n]$,
as the memory state $\memvar'$ mapping each address to the same value
as $\sigma$, except for $i$,
bound to $n$.
Formally,
$\setmem{\mem}{i}{n}$ is defined by the following rules:
\begin{eqnarray}
  \forall \memvar \in \Memvar, \locval_i \in \Locval, n \in \Natset, \locval_j \in \Locval.\ 
  i = j \Rightarrow \sigma[i/n] (j) = n,\\
  \forall \memvar \in \Memvar, \locval_i \in \Locval, n \in \Natset, \locval_j \in \Locval.\ 
  i  \neq j \Rightarrow \sigma[i/n] (j) = \memvar(j).
\end{eqnarray}

Let $\Memcom$ be the set of functions 
$\memcom:\Loccom \to \Com$,
called \emph{procedure environments},
mapping program names to commands (defined below),
and let $\memcom,\memcom_1, ...$  denote metavariables ranging over $\Memcom$.
We write $\body{\loccom}{\memcom}$ to refer to $\memcom(y)$,
the commands (or \emph{body}) of procedure $y$  in a given procedure environment $\memcom$.
An example of a procedure environment $\memcomsum$ is given
in Fig.~\ref{fig:ex-contract-rp2}, where $\body{\ysum}{\memcomsum}=\csum$.

\emph{Preconditions} (or \emph{assertions})
are predicates of arity one, taking as parameter a memory state
and returning an equational first-order logic formula.
Let metavariables $P, P_1, ...$ range over the set $\Precondition$ of preconditions.
For instance, using $\lambda$-notation, precondition
$P$ assessing that location $x_3$ is bound to $2$
can be defined by 
$P \triangleq \lambda \memvar. \memvar(3) = 2.$
This form will be more convenient
for relational properties (than {e.g.} $x_3 = 2$)
as it makes explicit
the memory states on which a property is evaluated.

\emph{Postconditions}
are predicates of arity two, taking as parameters two memory states
and returning an equational first-order logic formula.
Its two arguments refer to the initial and the final state. 
For instance, 
postcondition
$Q$ assessing that location $x_1$ was incremented 
(that is, $x_1=\old(x_1)+1$) can be defined in $\lambda$-notation by 
$Q \triangleq \lambda \memvar\memvar'.\, \memvar'(1)=\memvar(1)+1.$
Let metavariables $Q, Q_2, ... $ range over the set $\Postcondition$ of postconditions.

Finally, we define the set $\Memann$ of \emph{contract environments} $\memann: \Loccom  \to \Precondition \times \Postcondition$,
and metavariables $\memann,\memann_1,...$ to range over $\Memann$.
More precisely,  
$\memann$ maps a procedure name $y$ to the associated
(procedure) \emph{contract} $\memann(y)=(\pre{\loccom}{\memann},\post{\loccom}{\memann})$, composed of a pre- and a postcondition
for procedure $y$. 
As usual, a procedure contract will allow us
to specify the behavior of a single procedure call, that is,
if we start executing $\loccom$ in a memory state satisfying $\pre{\loccom}{\memann}$, and the evaluation terminates, the pair composed of the initial and final states will satisfy $\post{\loccom}{\memann}$.

\vspace{-3mm}
\subsection{Syntax for Expressions and Commands}
\vspace{-2mm}
\label{sec:exp-commands}

Let 
$\Aexp$, $\Bexp$ and $\Com$ denote respectively the sets of arithmetic expressions,
Boolean expressions and commands.
We denote by $\aexp,\aexp_1, ...$;
$\bexp,\bexp_1, ...$ and
$\com,\com_1, ...$ metavariables ranging, respectively, over those sets.
Syntax of arithmetic and Boolean expressions is given in Fig.~\ref{fig:exps-coms}.
Constants are natural numbers or Boolean values. 
Expressions use standard arithmetic, 
comparison and logic binary operators, denoted respectively $\opa ::= \{+, \times, - \}$, $\opb ::= \{\leqslant, =,\dots \}$, $\opl ::= \{\lor, \land\}$.
Since we use natural values, the
subtraction is bounded by 0, as in \coq: if $n'>n$, the result of $n-n'$ is considered to be 0.
Expressions also include locations, possibly with a dereference or address operators.

\begin{figure}[tb]
\begin{minipage}{3cm}
  \begin{align*}
    \aexp :&:= \nat & \text{natural const.}\\
           &|\ \locval & \text{location}\\
           &|\ *\locval  & \text{dereference}\\
           &|\ \&\locval & \text{address}\\
           &|\ \aexp_1 \opa \aexp_2& \text{arithm. oper.}
  \end{align*}
  \begin{align*}
    \bexp :&:= true\ |\ false & \text{Boolean const.}\\
           &|\ \aexp_1 \opb \aexp_2 & \text{comparison}\\
           &|\ \bexp_1 \opl \bexp_2\ |\ \lnot \bexp_1 & \text{logic oper.}
  \end{align*}
\end{minipage}
\hspace{3mm}
\begin{minipage}{7cm}
\begin{align*}
  \com :&:= \wskip & \text{do nothing} \\
    &|\ \locval := \aexp  & \text{\!\!\!\!\!\!\!\!\!\!\!\!\!\!direct assignment}\\
    &|\,*\locval := \aexp  & \text{\!\!\!\!\!\!\!\!\!\!\!\!\!\!\!\!\!\!\!\!\!indirect assignment}\\
    &|\ \com_1;\com_2  & \text{sequence}\\
    &|\ \wassert{P}  & \text{assertion}\\
    &|\ \wif{\bexp}{\com_1}{\com_2}  & \text{condition}\\
    &|\ \wwhilei{\bexp}{P}{\com_1}  & \text{loop}\\
    &|\ \wcall{\loccom}  & \text{\!\!\!\!\!\!\!\!\!\!\!\!\!\!procedure call}
\end{align*}
\end{minipage}
\vspace{-1mm}
\caption{Syntax of arithmetic and Boolean expressions and commands in $\langname$.}
\vspace{-5mm}
\label{fig:exps-coms}
\end{figure}

Figure~\ref{fig:exps-coms} also presents the syntax of commands in $\langname$.
Sequences, skip and conditions are standard.
An assignment can be done to a location directly or after a dereference.
Recall that a location $\locval_i$ contains as a value a natural number, 
say $v$, that can be seen in turn as the address of a location, namely $\locval_{v}$,
so the assignment $*\locval_i := \aexp$ writes the value of expression $a$ 
to the location $\locval_{v}$, 
while the address operation $\&\locval_i$ computes the address $i$ of $\locval_i$. 
An assertion command $\wassert{P}$ indicates that an assertion $P$ should
be valid at the point where the command occurs.
The loop command $\wwhilei{\bexp}{P}{\com_1}$ is always annotated with an invariant $P$.
As usual, this invariant should hold 
when we reach the command and be preserved by each loop step.
Command $\wcall{\loccom}$ is a procedure call. 
All annotations (assertions,  loop 
invariants and procedure contracts) will be ignored 
during the program execution 
and will be relevant only for program verification in Section~\ref{sec:verif-cond-gener}.
Procedures do not have
explicit parameters and return values (hence we use
the term \emph{procedure call} rather than \emph{function call}).
Instead, as in assembly code%
~\cite{Irvine:2014:ALX:2655333}, parameters and
return value(s) are shared implicitly between the caller and the callee through memory locations:
the caller must put/read the right values at the right locations before/after
the call. 
Finally, to avoid ambiguity, we regroup sequences of commands with $\{\,\}$.


\begin{figure}[t]
\centering
\begin{minipage}{2.5cm}
\begin{align*}
  \EAexpf{\nat}{\mem} & \triangleq \nat
\end{align*}
\end{minipage}
\begin{minipage}{3cm}
	\begin{align*}
	\EAexpf{\locval_i}{\mem} & \triangleq \mem(i)
	\end{align*}
\end{minipage}
\begin{minipage}{3cm}
	\begin{align*}
	\EAexpf{*\locval_i}{\mem} & \triangleq \mem(\mem(i))
	\end{align*}
\end{minipage}
\begin{minipage}{3cm}
	\begin{align*}
	\EAexpf{\&\locval_i}{\mem} & \triangleq i
	\end{align*}
\end{minipage}
\vspace{-2mm}
\caption{Evaluation of expressions in $\langname$ (selected rules).}
\vspace{-5mm}
\label{fig:exp-eval}
\end{figure}

\begin{figure}[t]
\centering
\scriptsize
\begin{minipage}{150mm}
\hspace{-2mm}
\begin{minipage}{25mm}
	\begin{equation*}
	\inferrule*
	{}
	{\eval{\wassert{P}}{\memvar}
		{\memcom}{\memvar}}
	\end{equation*}
\end{minipage}
\begin{minipage}{30mm}
  \begin{equation*}
    \inferrule*
    {\EAexpf{\aexp}{\memvar} = \nat}
    {\eval{\locval_i := \aexp}{\memvar}{\memcom}{\memvar[i/n]}}
  \end{equation*}
\end{minipage}
\begin{minipage}{35mm}
  \begin{equation*}
    \inferrule*
    {\EAexpf{\aexp}{\memvar} = \nat}
    {\eval{*\locval_i := \aexp}{\memvar}
      {\memcom}\memvar[\memvar(i)/n]}
  \end{equation*}
\end{minipage}
%
\begin{minipage}{33mm}
\begin{equation*}
  \inferrule*
  {\eval{\body{\loccom}{\memcom}}{\memvar_1}{\memcom}\memvar_2}
  {\eval{\wcall{\loccom}}{\memvar_1}{\memcom}\memvar_2}
\end{equation*}
\end{minipage}
\end{minipage}
\vspace{-3mm}
\caption{Operational semantics of commands in $\langname$ (selected rules).}
\vspace{-5mm}
\label{fig:oper-sem}
\end{figure}

\vspace{-5mm}
\subsection{Operational Semantics}
\vspace{-2mm}
\label{sec:oper-semant}


Evaluation of arithmetic and Boolean expressions in $\langname$ is defined by functions
$\EAexp$ and $\EBexp$.
Selected evaluation rules for arithmetic expressions are shown in Fig.~\ref{fig:exp-eval}.
Operations $*\locval_i$ and $\&\locval_i$ have a
semantics similar to the C language, {i.e.} dereferencing and address-of.
Semantics of Boolean expressions is standard~\cite{DBLP:books/daglib/0070910}.

Based on these evaluation functions, we can define the operational semantics
of commands 
in a given procedure environment $\memcom$.
Selected evaluation rules\footnote{For convenience of the reviewers,
full versions of Fig.~\ref{fig:exp-eval},~\ref{fig:oper-sem} are given in Appendix~\ref{app-sem}.}
are shown in Fig.~\ref{fig:oper-sem}.
As said above, both assertions and loop invariants can be seen
as program annotations that
do not influence the execution of the program itself.
Hence, command $\wassert{P}$ is equivalent to a skip. Likewise,
loop invariant $P$ has no influence
on the semantics of $\wwhilei{\bexp}{P}{\com}$.

We write $\Vdash\eval{\com}{\memvar}{\memcom}{\memvar'}$
to denote that $\eval{\com}{\memvar}{\memcom}{\memvar'}$
can be derived from the rules of Fig.~\ref{fig:oper-sem}.
Our \coq formalization,
inspired by \cite{SF},
provides a deep embedding of \langname,
with an associated parser, in  files
\texttt{Aexp.v}, \texttt{Bexp.v} and \texttt{Com.v}.

\vspace{-3mm}
\section{Functional Correctness}
\vspace{-2mm}
\label{sec:funct-corr}
We define functional correctness in a similar way to the original
\emph{Hoare triple} definition~\cite{DBLP:journals/cacm/Hoare69},
except that we also need a procedure environment $\memcom$, leading
to a quadruple denoted $\hoared{P}{\com}{Q}{\memcom}$. We will however
still refer by the term ``Hoare triple'' to the corresponding
program property, formally defined as follows.

\begin{definition}[Hoare triple]
\label{def:funct-corr-1}
Let $\com$ be a command, $\memcom$ a procedure environment,
and $P$ and $Q$ two assertions. We define a Hoare triple $\hoared{P}{\com}{Q}{\memcom}$
as follows:
\[
\hoared{P}{\com}{Q}{\memcom}\,
\triangleq\,
  \forall \memvar, \memvar' \in \Memvar.\  \Pre{P}{\memvar} \land
  (\Vdash\eval{\com}{\memvar}{\memcom}{\memvar'}) \Rightarrow \Post{Q}{\memvar}{\memvar'}.
\]

\end{definition}

\begin{figure}[tb]
	\begin{minipage}{4cm}
		\vspace{4mm} 
		Procedure environment:	
	\end{minipage}	
	\begin{minipage}{8cm}
		\begin{align*}
		\memcomsum \triangleq \left\{
		\ysum  \to \csum
		\right\}
		\end{align*}
	\end{minipage}
	\begin{minipage}{4cm}
		\vspace{4mm} 
		Hoare triple $\RelPII$:	
	\end{minipage}	
	\begin{minipage}{8cm}
		\begin{equation*}
		\memcomsum:\left\{
		\begin{array}{l@{}l}
		\True
		\end{array}
		\right\}
		\csum
		\left\{
		\begin{array}{l@{}l}
		\old(\locval_1) \geqslant \old(\locval_2) \Rightarrow \old(\locval_3) = \locval_3
		\end{array}
		\right\}
		\end{equation*}	
	\end{minipage}
	\begin{minipage}{2cm}
		\vspace{4mm} 
		Relational property $\RelPIII$:	
	\end{minipage}	
	\hspace{4mm}
	\begin{minipage}{8cm}
		\begin{equation*}
		\memcomsum:\left\{
		\begin{array}{l@{}l}
		\ben{\locval_1}{2} & < \ben{\locval_2}{2}\enskip\land \\
		\ben{\locval_2}{1} &= \ben{\locval_2}{2} \enskip\land \\
		\ben{\locval_1}{1} &= \ben{\locval_1}{2} + 1 \enskip\land\\
		\ben{\locval_3}{1} &= \ben{\locval_3}{2}+ \ben{\locval_1}{2}
		\end{array}
		\right\}
		\ben{\csum}{1}
		\sim
		\ben{\csum}{2}
		\left\{
		\begin{array}{l@{}l}
		\ben{\locval_3}{1} &= \ben{\locval_3}{2}
		\end{array}
		\right\}
		\end{equation*}
	\end{minipage}
	\vspace{-2mm}
	\caption{A procedure environment $\memcomsum$ associating procedure name $\ysum$
		with its body  $\csum$ (see Fig.~\ref{fig:ex-intro}),
		a Hoare triple $\RelPII$ for command $\csum$, 
		and a relational property $\RelPIII$ linking two runs of $\csum$.
	}
	\vspace{-5mm}
	\label{fig:ex-contract-rp2}
\end{figure}

Informally, our definition states
that, for a given $\memcom$, if a state $\memvar$ satisfies $P$ and
the execution of $\com$ on $\memvar$ terminates in a state $\memvar'$,
then $(\memvar,\memvar')$ satisfies $Q$.

\begin{example}
\label{ex:hoare-triple-example}
Figure~\ref{fig:ex-contract-rp2} gives an example of a Hoare triple
denoted $\RelPII$.
\hfil\qed
\end{example}


Next, we introduce notation $\hoarectx{\memcom}{\memann}$
to denote the fact that, 
for the given $\memann$ and  $\memcom$,
every procedure 
satisfies its contract.

\begin{definition}[Contract Validity]
Let $\memcom$ be a procedure environment and $\memann$ a contract environment.
We define contract validity $\hoarectx{\memcom}{\memann}$ as follows:
\[\hoarectx{\memcom}{\memann} \,\triangleq\,
  \forall \loccom \in \Loccom.\
  \hoared{\pre{\loccom}{\memann}}{\wcall{\loccom}}{\post{\loccom}{\memann}}{\memcom}).
\]
\end{definition}

The notion of contract validity is at the heart of
modular verification, since it allows assuming that the contracts
of the callees are satisfied during the verification of a Hoare triple.
More precisely, to state the validity of procedure contracts
without assuming anything about their bodies in our formalization,
we will consider an arbitrary choice of implementations $\memcom'$ of procedures
that satisfy the contracts, like in the first assumption of 
Theorem \ref{hoare:recursion} below. 
This theorem, 
taken from \cite[Th. 4.2]{DBLP:series/txcs/AptBO09} and reformulated for \langname
in~\cite{BlatterKPLiFM22}, states that
$\hoared{P}{\com}{Q}{\memcom}$ holds if
we can prove the contract of (the bodies 
in $\memcom$ of) all procedures in an arbitrary environment $\memcom'$
respecting the contracts, and if
the validity of contracts of $\memann$ for $\memcom$ 
implies the Hoare triple itself.
This theorem is the basis for modular verification of Hoare Triples,
as done for instance in Hoare Logic \cite{DBLP:journals/cacm/Hoare69,DBLP:books/daglib/0070910}
or verification condition generation. 

\begin{theorem}[Recursion]
  \label{hoare:recursion}
  Given a procedure environment $\memcom$
  and a contract environment $\memann$ such that 
  the following two assumptions hold:
  \begin{gather*}
    \forall \memcom' \in \Memcom.\ 
    \hoaredctx{\pre{\loccom}{\memann}}
    {\body{\loccom}{\memcom}}
    {\post{\loccom}{\memann}}{\forall \loccom \in \Loccom, \memcom'}{\hoarectx{\memcom'}{\memann}}, \\ 
    \hoaredctx{P}{\com}{Q}{\memcom}{\hoarectx{\memcom}{\memann}},
  \end{gather*}
  we have \enskip
  $
    \hoared{P}{\com}{Q}{\memcom}.
  $
\end{theorem}

We refer the reader to the \coq development,
more precisely the results
\verb|recursive_proc| and
\verb|recursive_hoare_triple| in file
\verb|Hoare_Triple.v| for a complete proof of
Theorem~\ref{hoare:recursion}.

\vspace{-3mm}
\section{Relational Functional Correctness}
\vspace{-2mm}
\label{sec:rel-prop}
Relational properties can be seen as an extension of
Hoare triples.
But, instead of linking one program with two properties, the pre- and
postconditions, relational
properties link $n$ programs to two properties, called \emph{relational precondition} and \emph{relational postcondition}.
A \emph{relational precondition} or \emph{assertion} (resp., 
\emph{relational postcondition}) for $n$ programs is
a predicate taking a sequence of $n$ (resp., $2n$)
memory states and returning a first-order logic formula.
Metavariables $\rela{P}$, $\rela{P'},\dots$ 
(resp., $\rela{Q}$, $\rela{Q'}$, $\dots$)
range over the corresponding sets.
As a simple example, the relational postcondition of $\RelPI$ 
(written in Fig.~\ref{fig:ex-intro} in Benton's notation)
can be stated in $\lambda$-notation as follows: \
$\lambda \memvar_1,\memvar_2,\memvar'_1,\memvar'_2\,\,.\,\, \memvar'_1(3) = \memvar'_2(3)$.

A \emph{relational property} is a property about $n$
programs $c_1, ..., c_n$, stating that if each program $c_i$
starts in a state $\memvar_i$ and ends in a state $\memvar'_i$ such
that $\rassert{P}{\memvar_1, ..., \memvar_n}$ holds,
then $\rassert{Q}{\memvar_1, ..., \memvar_n,\memvar'_1, ..., \memvar'_n}$ holds,
where $\rela{P}$ is a relational precondition 
and $\rela{Q}$ is a relational postcondition.
We formally define relational
correctness similarly to functional correctness (cf. Def. \ref{def:funct-corr-1}),
except that we now use sequences of commands and memory states.
We abbreviate by  $\sequence{u}{\nat}$
a sequence of elements $(u_k)^\nat_{k=1}=(u_1,\dots,u_n),$  
where
$k$ ranges from $1$ to $\nat$.
If $\nat \le 0$, $\sequence{u}{\nat}$ is the empty sequence denoted $[\ ]$.
If $\nat=1$, $(u)^1$ is the singleton sequence $(u)$.

\begin{definition}[Relational Hoare Triple] \label{def:rela}
Let $\memcom$ be a procedure environment,
$\sequence{\com}{\nat}$  a sequence of $\nat$ commands ($\nat \in \Natset^*$),
$\rela{P}$ and $\rela{Q}$ relational pre- and postcondition for $\nat$ commands.
The relational correctness of $\sequence\com{n}$ with respect to $\rela{P}$ and
$\rela{Q}$, denoted
$\rhoared{P}{\sequence\com{n}}{Q}{\memcom}$, is defined as follows:
\begin{gather*}
\rhoared{P}{\sequence\com{n}}{Q}{\memcom} \triangleq\\
\forall \sequence{\memvar}{\nat}, \sequencep{\memvar}{\nat}.\
    \rPre{P}{\sequence{\memvar}{\nat}}
     \land
    (\evalrr{\com}{\memvar}{\memvar'}{\psi}{\nat})
    \Rightarrow
    \rPost{Q}{\sequence{\memvar}{\nat}}{\sequencep{\memvar}{\nat}}.
\end{gather*}
\end{definition}

For $n=1$, this notion defines a Hoare triple. 
It also generalizes Benton's notation~\cite{DBLP:conf/popl/Benton04}
for two commands:
$  \rhoared{P}{c_1\sim c_2}{Q}{\memcom}$.
As Benton's work mostly focused on comparing equivalent programs, 
using symbol $\sim$
was quite natural.

\begin{example}
\label{ex:rel-hoare-triple-example}
Relational property $\RelPIII$ introduced 
in Ex.~\ref{ex:rel-contract} is formalized (in Benton's notation)
in Fig.~\ref{fig:ex-contract-rp2}.
Below, we will illustrate modular verification of relational properties
by deducing $\RelPI$ from $\RelPIII$
and partial contract $\RelPII$ of $\csum$.
\hfil\qed
\end{example}


We will now extend Theorem~\ref{hoare:recursion} to relational contract
environments. 
A \emph{relational contract environment} $\rmemann$ 
maps a sequence of program names $\sequence\loccom\nat$
to a \emph{relational contract}, composed of a relational pre- and postcondition,
denoted 
$\rmemann(\sequence\loccom\nat)=(\,\rpre{\sequence\loccom\nat}{\rmemann},
\,\rpost{\sequence\loccom\nat}{\rmemann}\,)$.
Practical applications require only a finite number of properties, so
the relational contract can be assumed trivial for all 
except a finite number of sequences.
A relational contract environment generalizes a contract environment, since 
a standard procedure contract is a relational contract 
(for a sequence of exactly one element).
Notice that $\rmemann$ considers only one relational property for 
a given sequence $\sequence\loccom\nat$: this is not a limitation 
since several properties can be 
encoded in one contract.
We define the set of 
relational contract environments $\RMemann$,
and metavariables $\rmemann,\rmemann_0,\rmemann_1,...$ will range over $\RMemann$.

We introduce notation $\rhoarectx{\memcom}{\rmemann}$
to denote the fact that all procedures defined in $\memcom$ satisfy the relational contracts
in which they are involved in
$\rmemann$.
\begin{definition}[Relational Contract Validity]
Let $\memcom$ be a procedure environment and $\rmemann$ a relational
contract environment. We define $\rhoarectx{\memcom}{\rmemann}$ as follows:\\
\begin{adjustbox}{width=\textwidth}
\begin{minipage}{13.2cm}
\vspace{-2mm}	
\begin{gather*}
\rhoarectx{\memcom}{\rmemann}\triangleq
\forall \sequence\loccom\nat\in\Domf\rmemann,
  \,\nat>0 \Rightarrow
\hoared{\rpre{\sequence{\loccom}{\nat}}{\rmemann}}
{(\wcall{\loccom_k})_{k=1}^{\nat}}
{\rpost{\sequence{\loccom}{\nat}}{\rmemann}}{\memcom}.
\end{gather*}
\end{minipage}
\end{adjustbox}
\end{definition}

\begin{theorem}[Relational Recursion]
	\label{rela:recursion}
	Given a  procedure environment $\memcom$
	and a relational contract environment $\rmemann$ such that the following two assumptions hold:
	\[
	\begin{array}{c}
          \forall \memcom' \in \Memcom. \enskip 
	\hoaredctxTwoLines{\rpre{\sequence\loccom\nat}{\rmemann}}
	{(\body{\loccom_k}{\memcom})_{k=1}^{\nat}}
          {\rpost{\sequence\loccom\nat}{\rmemann}}
          {\forall \sequence\loccom\nat\in\Domf\rmemann,\memcom'}{\rhoarectx{\memcom}{\rmemann}},
	\end{array}
	\]
	\begin{equation*}
	\rhoaredctx{P}{\sequence\com\nat}{Q}{\memcom}{\rhoarectx{\memcom}{\rmemann}}
	\end{equation*}
	then we have \enskip
$	
	\rhoared{P}{\sequence\com\nat}{Q}{\memcom}.
$	
\end{theorem}

The Coq proof (which is a straightforward extension of the proof of
Theorem~\ref{hoare:recursion}) is available in {\tt Rela.v}, Theorem~\texttt{recursion_relational}.

\vspace{-3mm}
\section{Optimized Verification Condition Generator}
\vspace{-2mm}
\label{sec:verif-cond-gener}
A standard way~\cite{Floyd1967} for verifying that a Hoare triple holds is to use
a verification condition generator (VCGen).
In this section, we formalize a VCGen for Hoare triples
such that if all verification conditions that it generates
are valid, then the Hoare triple is valid according to Def.~\ref{def:funct-corr-1}.
The VCGen described in this section
is based on optimizations introduced in \cite{DBLP:conf/popl/FlanaganS01}.
Such optimizations allow the VCGen to return
formulas whose size is linear with respect to the size of the program itself, and are
now part of any state-of-the-art deductive verification tool.
The key idea is to avoid splitting verification condition generation into two separated
sub-generation at each conditional.
The definition is formalized in \coq in the file
{\tt Vcg_Opt.v}, where we also prove that the verification conditions
of this optimized VCGen imply those of the naive VCGen presented in~\cite{BlatterKPLiFM22}.
This will allow us to use the optimized VCGen (or more generally any VCGen satisfying the
properties stated in Theorem~\ref{hoare:proof} below)
for the verification of relational properties as well (see Section~\ref{sec:relat-prop-verif}).

\vspace{-2mm}
\subsection{Verification Condition Generator}
\label{vcgdef}
\vspace{-2mm}
When defining the naive VCGen in~\cite{BlatterKPLiFM22},
we proposed a modular definition.
Namely, we divided it into
three functions $\Tcn$, $\Tan$ and $\Tfn$. Here, we follow the same
approach for the optimized VCGen, using three new functions $\Tc$, $\Ta$, and $\Tf$:

\begin{itemize}
\item function $\Tc$ generates the main verification condition,
expressing that the postcondition holds in the final state,
assuming auxiliary annotations hold;
\item function $\Ta$ generates auxiliary verification conditions stemming from
  assertions, loop invariants, and preconditions of called procedures;
\item finally, function $\Tf$ generates verification conditions for the
auxiliary procedures that are called by the main program, to ensure that
their bodies respect their contracts.
\end{itemize}

\begin{figure}[tb]
\centering
\begin{minipage}{12cm}
\begin{align*}
  \Tcf{\wskip}{\mem}{\mem'}{\memann}{f} & \triangleq f (\mem = \mem')\\
  \Tcf{\locval_i := \aexp}{\mem}{\mem'}{\memann}{f} &
  \triangleq
  f (\mem' = set(\mem,i,\EAexpf{\aexp}{\mem}))\\
  \Tcf{*\locval_i := \aexp}{\mem}{\mem'}{\memann}{f} &
  \triangleq
  f (\mem' = set(\mem,\mem(i),\EAexpf{\aexp}{\mem}))\\
  \Tcf{\wassert{P}}{\mem}{\mem'}{\memann}{f} & \triangleq
  f (\assert{P}{\mem}\land\mem = \mem')\\
  \begin{split}
  \Tcf{\com_0;\com_1}{\mem}{\mem'}{\memann}{f} &
  \triangleq
  \forall \mem'',
  \Tc\llbracket \com_0 \rrbracket (\mem,\mem'',\memann,\lambda p_1.\\
  &  \Tc\llbracket \com_1 \rrbracket (\mem'',\mem',\memann,\lambda p_2.
  f (p_1 \land p_2)))
  \end{split}\\
  \begin{split}
  \Tcf{\wif{\bexp}{\com_0}{\com_1}}{\mem}{\mem'}{\memann}{f} &
  \triangleq
  \Tc\llbracket \com_0 \rrbracket (\mem,\mem',\memann,\lambda p_1.\\
   &  \Tc\llbracket \com_1 \rrbracket (\mem,\mem',\memann,\lambda p_2. \\
   &  f ((b \equiv \True \Rightarrow p_1) \land (b \equiv \False \Rightarrow p_2))))
   \end{split}\\
  \Tcf{\wcall{\loccom}}{\mem}{\mem'}{\memann}{f}  &
  \triangleq
  f (\pre{\loccom}{\memann}(\mem) \land \post{\memann}{\loccom}(\mem,\mem'))\\
  \Tcf{\wwhilei{\bexp}{inv}{\com}}{\mem}{\mem'}{\memann}{f} &
  \triangleq
  f (inv\ \mem \land inv\ \mem' \land \neg(\EBexpf{\bexp}{\mem'}))
\end{align*}
\end{minipage}
\vspace{-2mm}
\caption{Definition of function $\Tc$ generating the main verification condition.
}
\label{fig:def-Tc}
\vspace{-5mm}
\end{figure}

\begin{definition}[Function $\Tc$ generating the main verification condition]
  \label{def:tc-fun}
  Given a command $\com$, two memory states $\sigma$ and $\sigma'$,
  a contract environment $\memann$, and a function
  $f$ taking a formula as argument and returning a formula,
  function $\Tc$ returns a formula defined by case analysis on $\com$ as shown in
  Fig.~\ref{fig:def-Tc}.
\end{definition}

State $\sigma$ represents the state before executing the command, while
$\sigma'$ represents the state after it.
Intuitively, the argument that gets passed to $f$ is the formula
that relates $\sigma$ and $\sigma'$ according to $\com$ itself. Thus, if $f$ is
of the form $\lambda p. p\Rightarrow Q(\sigma,\sigma')$, as in Theorem~\ref{hoare:proof}
below, the resulting formula is a verification condition for post-condition $Q$ to hold.

For $\wskip$, which does nothing, both
states are identical. For assignments, $\mem'$ is simply the update of $\mem$. An assertion
introduces a hypothesis over $\mem$ but leaves it unchanged.
For a sequence, a fresh memory state $\mem''$ is
introduced, and we compose the VCGen.
For a conditional,
if the condition evaluates to \True, we select the condition from the
\emph{then} branch, and otherwise from the \emph{else} branch.
Note that, contrary to the naive VCGen, we perform a single call to $f$,
ensuring the linearity of the formula.

The rule for calls simply assumes that before the call $\mem$ satisfies 
$\pre{y}{\phi}$ and after the call $\mem$ and $\mem'$ satisfy
$\post{y}{\phi}$.
Finally,
$\Tc$ assumes that, for a loop, both the initial state $\mem$ and the final one $\mem'$ satisfy
the loop invariant. Additionally, in $\mem'$ the loop condition evaluates to \False.
As for an assertion, the callee's precondition and the loop invariant
are just assumed to be true; function
$\Ta$, defined below, generates the corresponding proof obligations.

\begin{example}
\label{ex:Tc}
For $\com \triangleq \wif{\False}{\wskip}{\locval_1:= 2}$ we have:
\begin{gather*}
  \Tcf{\com}{\mem}{\mem'}{\memann}{\lambda p.\ p \Rightarrow  \mem'(1) = 2} \equiv \\
  (\False \equiv \True \Rightarrow \mem = \mem') \land
  (\False \equiv \False \Rightarrow \mem' = \setmem{\mem}{1}{2}) \Rightarrow
  \mem'(1) = 2. \quad \qed
\end{gather*}
\end{example}

Lemma \ref{lemma:tcoptnaive} establishes a relation between functions
$\Tc$ and $\Tcn$: the formulas generated by $\Tc$
imply the formulas generated by $\Tcn$.

\begin{lemma}
  \label{lemma:tcoptnaive}
  Given a program $\com$, a procedure contract environment $\memann$,
  a memory state $\mem$ and an assertion $P$, if we have \ 
  $ 
    \forall \mem' \in \Sigma,\,\,
    \Tcf{\com}{\mem}{\,\mem'}{\,\memann}{\,\lambda p.\, p \Rightarrow P(\mem')},
  $ 
  then we have \
  $ 
    \Tcnf{\com}{\mem}{\memann}{P}.
  $ 
\end{lemma}

\begin{proof}
  By structural induction over $\com$.
  \qed
\end{proof}


\begin{figure}[tb]
\centering
\begin{minipage}{12cm}
\begin{align*}
  \Taf{\wskip}{\mem}{\memann} & \triangleq \True\\
  \Taf{\locval := \aexp}{\mem}{\memann} & \triangleq \True\\
  \Taf{*\locval := \aexp}{\mem}{\memann} & \triangleq \True\\
  \Taf{\wassert{P}}{\mem}{\memann} & \triangleq \assert{P}{\mem}\\
  \begin{split}
    \Taf{\com_0;\com_1}{\mem}{\memann} & \triangleq  \Taf{\com_0}{\mem}{\memann} \,\,\land \\
    & \hspace{-8mm} \forall \mem',\,\,\Tcf{\com_0}{\mem}{\,\mem'}{\,\memann}
    {\,\lambda p.\, p \Rightarrow \Taf{\com_1}{\mem'}{\memann}}
  \end{split}\\
  \begin{split}
    \Taf{\wif{\bexp}{\com_0}{\com_1}}{\mem}{\memann} & \triangleq
    (\EBexpf{\bexp}{\mem'} \Rightarrow \Taf{\com_0}{\mem}{\memann}) \,\,\land \\
    & \hspace{-8mm} (\neg(\EBexpf{\bexp}{\mem'}) \Rightarrow \Taf{\com_1}{\mem}{\memann})
  \end{split}\\
  \Taf{\wcall{\loccom}}{\mem}{\memann} & \triangleq \pre{\loccom}{\memann}(\mem)\\
  \begin{split}
  \Taf{\wwhilei{\bexp}{inv}{\com}}{\mem}{\memann} & \triangleq  inv(\mem) \,\land \\
                 & \hspace{-8mm} (\forall \mem', \,\,inv(\mem')  \,\Rightarrow\,
                           \EBexpf{\bexp}{\mem'} \Rightarrow \Taf{c}{\mem'}{\memann}) \,\,\land \\
                 & \hspace{-8mm} (\forall \mem' \mem'',\,\, inv(\mem')  \,\Rightarrow\,
                     \Tcf{\com}{\mem'}{\mem''}{\memann}{\lambda p.\, p \Rightarrow inv(\mem''))}
 \end{split}
\end{align*}
\end{minipage}
\vspace{-2mm}
\caption{Definition of function $\Ta$ generating auxiliary verification conditions.}
\label{fig:def-Ta}
\vspace{-5mm}
\end{figure}

\begin{definition}[Function $\Ta$ generating the auxiliary verification condition]
  \label{def:ta-fun}
  Given a command $\com$, a memory state $\mem$ representing the
  state before the command, and a contract environment $\memann$,
  function $\Tan$ returns a formula defined by case analysis on $\com$
  as shown in Fig.~\ref{fig:def-Ta}.
\end{definition}

Basically, $\Ta$ collects all assertions, preconditions of called procedures, as
well as invariant establishment and preservation, and lifts the corresponding
formulas to constraints on the initial state $\mem$ through the use of $\Tc$.

As for $\Tc$, the formulas generated by $\Ta$
imply those generated by $\Tan$.

\begin{lemma}
  \label{lemma:taoptnaive}
  For a given program $\com$, a procedure contract environment $\memann$,
  and a memory state $\mem$, if we have \enskip
  $ 
    \Taf{\com}{\mem}{\memann},
  $ 
  \enskip then we have \enskip
  $ 
    \Tanf{\com}{\mem}{\memann}.
  $ 
\end{lemma}

\begin{proof}
  By structural induction over $\com$.
  \qed
\end{proof}

\smallskip
Finally, we define the function for generating the conditions
for verifying that the body of each procedure defined in $\memcom$
respects its contract defined in $\memann$.

\begin{definition}[Function $\Tf$ generating the procedure verification condition]
    $\Tf$ takes as argument two environments $\memcom$ and $\memann$ and returns
  a formula:
\begin{equation*}
  \begin{array}{ll}
  \Tff{\memann}{\memcom} \triangleq
  \forall \loccom, \mem, \mem'.\,\,\pre{\loccom}{\memann}(\mem) \,\Rightarrow & \\ &\hspace{-45mm} 
  \Taf{\body{\loccom}{\memcom}}{\mem}{\memann} \land 
   \Tcf{\body{\loccom}{\memcom}}{\mem}{\mem'}{\memann}
  {\lambda p. p \Rightarrow \post{\loccom}{\memann}(\mem,\mem')}.
  \end{array}
\end{equation*}
\end{definition}

Finally, the formulas generated by $\Tf$
imply those generated by $\Tfn$.

\begin{lemma}
  \label{lemma:tfoptnaive}
  For a given procedure environment $\memcom$, and a
  procedure contract environment $\memann$, if we have \enskip
  $ 
    \Tff{\memann}{\memcom},
  $ 
  \enskip then we have \enskip
  $ 
    \Tfnf{\memann}{\memcom}.
  $ 
\end{lemma}

\begin{proof}
  Using Lemmas~\ref{lemma:tcoptnaive} and~\ref{lemma:taoptnaive}.
  \qed
\end{proof}

The definition of the optimized VCGen and its link to the naive version
can be found in file {\tt Vcg_Opt.v} of the \coq development.

\vspace{-2mm}
\subsection{Hoare Triple Verification}
\label{hoareverif}
\vspace{-1mm}

Using the VCGen defined in Sec.~\ref{vcgdef}, we can state
the theorem establishing how a Hoare Triple can be verified.
The proof can be found in file
{\tt Correct.v} of the \coq development.

\begin{theorem}[Soundness of VCGen]
  \label{hoare:proof}
  Assume that 
we have
$\Tff{\memann}{\memcom}$ and
  \begin{gather*}
    \forall \mem.\,\,\assert{P}{\mem} \Rightarrow \Taf{\com}{\mem}{\memann}, \\
    \forall \mem, \mem'.\,\,\assert{P}{\mem} \Rightarrow
    \Tcf{\com}{\mem}{\,\mem'}{\,\memann} {\,\lambda p. \,p \,\Rightarrow\, Q(\mem,\mem')}.
  \end{gather*}
  Then we have
    $\hoared{P}{\com}{Q}{\memcom}$.
\end{theorem}

\begin{proof}
  By soundness of the naive VCGen \cite[Th.\,3]{BlatterKPLiFM22} and Lemmas~\ref{lemma:tcoptnaive}, \ref{lemma:taoptnaive}, \ref{lemma:tfoptnaive}.
\qed
\end{proof}

\vspace{-4mm}
\section{Modular Verification of Relational Properties}
\vspace{-2mm}
\label{sec:relat-prop-verif}
In this section, we propose a modular verification method for
relational
properties (defined in Section~\ref{sec:rel-prop})
using the optimized VCGen defined in Section~\ref{sec:verif-cond-gener}
(or, more
generally, any VCGen respecting Theorem~\ref{hoare:proof}).
First, we define the function $\Tro$ for the recursive call of 
$\Tc$ on a sequence of commands and memory states.

\begin{definition}[Function $\Tro$]
  \label{def:Trf}
  Given  a sequence of commands $\sequence\com\nat$ and
  a sequence of memory states $\sequence\mem\nat$,
  a contract environment
  $\memann$ and a function $f$ taking as argument a formula and returning a formula,
  function $\Tro$ is defined by induction on $\nat$
  for the basis ($\nat = 0$) and inductive case ($\nat \in \Natset^*$) as follows:
    \begin{gather*}
      \Trof{[\ ]}{[\ ]}{[\ ]}{\memann}{f} \triangleq f(True),\\
      \Trof{\sequence{\com}{\nat}}
      {\sequence{\mem}{\nat}}{\sequencep{\mem}{\nat}}{\memann}{f} \triangleq\\
      \Tcf{\com_\nat}{\mem_\nat}{\mem_\nat'}{\memann}
      {\,\,\lambda p_n.\
        \Trof{\sequence{\com}{\nat-1}}
        {\sequence{\mem}{\nat-1}}
        {\sequencep{\mem}{\nat-1}}{\memann}
        {\lambda p_{n-1}.\ f(p_n \land p_{n-1})}}.
    \end{gather*}
\end{definition}

Intuitively, like in Def.~\ref{def:tc-fun}, the argument that gets passed to $f$ is the formula that relates
the $\nat$ pre-states $\sequence{\mem}{\nat}$ to the $\nat$ post-states $\sequencep{\mem}{\nat}$ when all $\sequence{\com}{\nat}$ are executed. Again, if $f$ is of the form $\lambda p. p\Rightarrow \rPost{Q}{\sequence{\memvar}{\nat}}{\sequencep{\memvar}{\nat}}$, the resulting formula is a verification condition for the relational postcondition $\rela{Q}$ to hold.
More concretely, for $\nat=2$, and $f$ as above, we obtain:
\begin{gather*}
\Trof{(\com_1, \com_2)}{(\mem_1, \mem_2)}{(\mem_1', \mem_2')}{\memann}{\lambda p.p\Rightarrow  \rPost{Q}{(\mem_1,\mem_2)}{(\mem_1', \mem_2')}} \equiv \\
\Tcf{\com_2}{\mem_2}{\mem_2'}{\memann}
{\,\lambda p_2. \Tcf{\com_1}{\mem_1}{\mem_1'}{\memann}
	{\,\lambda p_1. p_2 \land p_1 \Rightarrow \rPost{Q}{(\mem_1,\mem_2)}{(\mem_1', \mem_2')}}}.
\end{gather*}

\smallskip
We similarly define a notation for the 
auxiliary verification conditions
for a sequence of $\nat$ commands. Basically, this is the
conjunction of the auxiliary verification conditions generated
by $\Ta$ on each individual command.
\begin{definition}[Function $\Taor$]
  \label{def:Tarf}
  Given a sequence of commands
  $\sequence{\com}{\nat}$ and a sequence of memory states $\sequence{\mem}{\nat}$,
  we define function $\Taor$ as follows:
  \begin{equation*}
    \Taorf{\sequence{\com}{\nat}}{\sequence{\mem}{\nat}}{\memann} \triangleq
    \bigwedge_{i=1}^n \Taf{\com_i}{\mem_i}{\memann}.
  \end{equation*}
\end{definition}

A standard contract over a single procedure $\loccom$ can be used directly whenever
there is a call to $\loccom$. For a relational contract over $\sequence{\loccom}{\nat}$, things are more complicated: there is not a single program point where we can apply the relational contract. Instead, we have to somehow track in the generated formulas all the calls that have been made, and to guard the application of the relational contract by a constraint stating that all the appropriate calls have indeed taken place. In order to achieve that, we start by defining a notation for the conjunction of a sequence of procedure calls and associated memory states:
\begin{definition}[Functions $\Proccall$ and $\Procpred$] 
  \begin{gather*}
      \Proccallf{\loccom}{\mem}{\mem'}{\memcom}
      \triangleq\ 
      \Vdash\eval{\wcall{\loccom}}{\mem}{\memcom}{\mem'},
   \\
    \Procpredf{\sequence{\loccom}{\nat}}{\sequence{\mem}{\nat}}{\sequencep{\mem}{\nat}}{\memcom}
    \triangleq
    \bigwedge_{i=1}^n \Proccallf{\loccom_i}{\mem_i}{\mem'_i}{\memcom}.
  \end{gather*}

\end{definition}

Then, we can define function $\Tpr$ translating relational contracts
into a logical formula, using $\Procpred$ to guard its application
with tracked calls.
\begin{definition}[Function $\Tpr$]
  \begin{gather*}
    \Tprf{\rmemann}{\memcom}  \triangleq\\
    \forall \sequence{\loccom}{\nat},\sequence{\mem}{\nat}, \sequencep\mem\nat,
    \,\, \nat>0 \,\,\Rightarrow\,\,
    \Procpredf{\sequence{\loccom}{\nat}}
    {\sequence{\mem}{\nat}}
    {\sequencep{\mem}{\nat}}{\memcom} \,\,\Rightarrow\,\,\\
    \rpre{\sequence{\loccom}{\nat}}{\rmemann}\sequence{\mem}{\nat} \,\,\Rightarrow\,\,
    \rpost{\sequence{\loccom}{\nat}}{\rmemann}\sequence\mem\nat \sequencep{\mem}\nat.
  \end{gather*}
\end{definition}

We now define function $\Phicall$ to lift a relational procedure contract 
with an associated tracked call predicate and reduce it to a standard contract.
\begin{equation*}
  \Phicallf{\rmemann}{\memcom} \triangleq
  \lambda  \loccom.
  (\lambda \mem. \rpre{(y)^1}{\rmemann} (\mem)^1,
  \lambda \mem \mem'. \rpost{(y)^1}{\rmemann}(\mem)^1 (\mem')^1
  \land \Proccallf{\loccom}{\mem}{\mem'}{\memcom}).
\end{equation*}


Finally, using function $\Tpr$ and $\Phicall$,
we can define function
$\Tfor$ for generating the verification condition for verifying
that the bodies of each sequence of procedures
respect the relational contract defined in $\rmemann$: thanks to $\Phicall$, each call instruction will result in a corresponding $\Proccall$ occurrence in the generated formula, so that it will be possible to make use of the relational contracts hypotheses in $\Tpr$ when the appropriate sequences of calls occur.
\begin{definition}[Function $\Tfor$]
  \begin{gather*}
    \Tforf{\rmemann}{\memcom} \triangleq \\
    \forall \sequence\loccom\nat,\sequence\mem\nat,\sequencep\mem\nat, \memcom',\,\,\,\,
    \rpre{\sequence{\loccom}{\nat}}{\rmemann} \,\,\Rightarrow\,\,
    \Tprf{\rmemann}{\memcom'} \,\,\Rightarrow\,\,\\
    \Taorf{(\body{\loccom_k}{\,\memcom})_{k=1}^{\nat}}
    {\,\sequence\mem\nat}{\,\Phicallf{\rmemann}{\memcom'}} \,\,\land\\
    \Trof{(\body{\loccom_k}{\memcom})_{k=1}^{\nat}}
    {\,\sequence\mem\nat}{\,\sequencep\mem\nat}{\,\Phicallf\rmemann{\memcom'}}
    {\,\lambda p. p \,\Rightarrow\, \rpost{\sequence{\loccom}{\nat}}{\rmemann})}.
  \end{gather*}
\end{definition}

Using functions $\Tro$, $\Taor$ and $\Tfor$, we can now give the main result
of this paper, i.e.
that the verification of relational
properties with the VCGen is correct.
\begin{theorem}[{Soundness of relational VCGen}]
  \label{rela:proof}
  For any sequence of commands $\sequence\com\nat$,
  contract environment $\rmemann$, procedure environment $\memcom$,
  and relational pre- and postcondition $\rela{P}$ and $\rela{Q}$,
  if the following three properties hold:
  \begin{gather}
    \label{hyp:1}
    \Tforf{\rmemann}{\memcom},\\
    \label{hyp:2}
    \forall \sequence\mem\nat,\memcom',\,\,\,\,
    \rassert{P}{\sequence\mem\nat} \land
    \Tprf{\rmemann}{\memcom} \,\,\Rightarrow\,\,
    \Taorf{\sequence\com\nat}{\,\sequence\mem\nat}{\,\Phicallf{\rmemann}{\memcom'}},\\
    \forall \sequence\mem\nat,\sequencep\mem\nat,\memcom',\,\,\,\,
    \rassert{P}{\sequence\mem\nat} \,\land\,
    \Tprf{\rmemann}{\memcom} \,\,\Rightarrow \,\, \nonumber\\
    \label{hyp:3}
    \Trof{\sequence\com\nat}{\sequence\mem\nat}
    {\,\sequencep\mem\nat}{\,\Phicallf{\rmemann}{\memcom'}}
    {\,\,\lambda p.\, p \Rightarrow \rassert{Q}{\sequence\mem\nat,\sequencep\mem\nat}},
  \end{gather}
  then we have \enskip
  $
    \rhoared{P}{\sequence\com\nat}{Q}{\memcom}.
  $
\end{theorem}

In other words, a relational property is valid
if all relational procedure contracts are valid, and, assuming the relational
precondition holds, both the 
auxiliary verification conditions
and the main relational verification condition hold.
The corresponding
\coq formalization is available in file \texttt{Rela.v},
and the \coq proof of Theorem~\ref{rela:proof} is in file \texttt{Correct_Rela.v}.

\begin{example}
  \label{ex:rel-prop-proof}
Consider $\memcom=\memcomsum$ and $\rmemann$ which encodes $\RelPII$ and $\RelPIII$.
The relational property $\RelPI$ of Fig.~\ref{fig:ex-intro} can now be proven valid
in a modular way, 
using $\RelPII$ and $\RelPIII$,
by the proposed technique based on Theorem~\ref{rela:proof} (see file {\tt Examples.v}
of the \coq development).
For instance, (\ref{hyp:3}) becomes the formula of Fig.~\ref{fig:ex-third-hypo-th}.
There, the relational precondition is given by (\ref{rpre}), while
the simplified (instantiated for sequence $({\ysum},{\ysum})$)
translation of the relational contracts $\Tprf{\rmemann}{\memcom}$
is given by (\ref{tprf}). Finally, (\ref{trof}) gives the
main verification condition:
\begin{gather*}
\Trof{(\csI,\csII)}{\,(\sigma_1,\sigma_2)}
{\,(\sigma_1',\sigma_2')}
{\,\Phicallf{\rmemann}{\memcom'}}{\,\lambda p. p \Rightarrow \mem_1'[3] = \mem_2'[3]}, \mbox{\,where\ }\Phicallf{\rmemann}{\memcom'}= \\
\{\ysum \rightarrow
(\lambda \mem.\, \True,
\lambda \mem,\mem'.\, \mem[1] \geqslant \mem[2] \Rightarrow \mem[3] = \mem'[3] \land
\Proccallf{\ysum}{\mem}{\mem'}{\memcom'})\}.
\end{gather*}
\noindent
Long for a manual proof, such formulas are well-treated by 
solvers.
\qed
\end{example}

\begin{figure}[tb]
\centering
\begin{minipage}{12cm}
\begin{gather}
  \forall \mem_1,\mem_2,\mem_1',\mem_2',\memcom.\nonumber\\
  \boxed{\mem_1(1) = \mem_2(1)} \label{rpre} \\
  \land\nonumber \\
  \boxed{
  \begin{gathered}
    (\forall \mem_1,\mem_2,\mem_1',\mem_2'.\\
    \Proccallf{\ysum}{\mem_1}{\mem_1'}{\memcom} \land
    \Proccallf{\ysum}{\mem_2}{\mem_2'}{\memcom} \land \\
    \mem_2(1) < \mem_2(2) \land \mem_1(2) = \mem_2(2) \land \\
    \mem_1(1) = \mem_2(1) + 1 \land
    \mem_1(3) = \mem_2(3) + \mem_2(1)\\
    \Rightarrow \\
    \mem_1'(3) = \mem_2'(3)
    )\end{gathered}
  }
  \label{tprf}\\
  \Rightarrow \nonumber\\
  \boxed{
    \begin{gathered}
      \forall \mem_1'',\mem_1''',\mem_2'',\mem_2'''. \\
      \mem_1'' = set(\mem_1,1,1) \land
      \mem_1''' = set(\mem_1'',3,0) \land \\
      ((\mem_1'''(1) \geqslant \mem_1'''(2) \Rightarrow \mem_1'''(3) = \mem_1'(3))
      \land \Proccallf{\ysum}{\mem_1'''}{\mem_1'}{\memcom}) \land \\
      \mem_2'' = set(\mem_2,1,0) \land
      \mem_2''' = set(\mem_2'',3,0) \land \\
      ((\mem_2'''(1) \geqslant \mem_2'''(2) \Rightarrow \mem_2'''(3) = \mem_2'(3))
      \land \Proccallf{\ysum}{\mem_2'''}{\mem_2'}{\memcom}) \\
      \Rightarrow \\
      \mem_1'(3) = \mem_2'(3)
    \end{gathered}
  }
  \label{trof}
\end{gather}
\end{minipage}
\vspace{-2mm}
\caption{Assumption (\ref{hyp:3}) of Theorem~\ref{rela:proof} illustrated for 
property $\RelPI$ of Fig.~\ref{fig:ex-intro}.}
\label{fig:ex-third-hypo-th}
\vspace{-5mm}
\end{figure}




\vspace{-5mm}
\section{Related Work}
\vspace{-2mm}
\label{sec:related}
\paragraph{Relational Property Verification.}

Significant work has been done on relational program verification 
(see~\cite{47rela,next700} for a detailed
state of the art). We discuss below some of the efforts the most closely related to our work.


Various relational logics have been designed as extensions
to Hoare Logic, such as Relational Hoare Logic~\cite{DBLP:conf/popl/Benton04}
and Cartesian Hoare Logic~\cite{SousaD16}.
As our approach, those logics consider
for each command a set of associated memory states
in the very rules of the system,
thus avoiding additional separation assumptions.
Limitations of these logics are often
the absence of support for aliasing
or a limited form of relational properties. For instance,
Relational Hoare Logic supports only relational properties with two commands
and Cartesian Hoare Logic supports only $k$-safety properties (relational properties on
the same command).
Our method has an advanced support of aliasing and
supports a very general definition of relational properties,
possibly between several dissimilar commands.

Self-compositon~\cite{DBLP:journals/mscs/BartheDR11,DBLP:conf/fm/SchebenS14,blatterKGP17}
and its derivations~\cite{DBLP:conf/fm/BartheCK11,ShemerCAV2021,DBLP:conf/esop/EilersMH18}
are well-known appro\-aches to deal with relational properties.
This is in particular due to their flexibility: self-composition
methods can be applied as a preprocessing step to different verification approaches.
For example, self-composition is used in combination with symbolic execution and model checking
for verification of voting functions~\cite{BeckertBormerEA2016}.
Other examples are the use of self-composition in combination with verification condition generation
in the context of the Java language~\cite{GuillaumeCAD20} or the C language~\cite{blatterKGP17,blatterKGPP18}.
In general, the support of aliasing of C programs
in these last efforts is very limited due the problems mentioned earlier.
Compared to these techniques,
where self-composition is applied before the generation of verification
conditions (and therefore requires taking care about separation of memory
states of the considered programs),
our method can be seen as relating the considered programs' semantics
directly at the level of the verification conditions,
where separation of their memory states is already ensured,
thus avoiding the need to take care of this separation explicitly.

%

Finally, another advanced approach for relational verification is the
translation of the relational problem into Horn clauses and
their proof using constraint solving%
~\cite{DBLP:journals/jar/KieferKU18,HiroshCAV2021}.
The benefit of constraint solving lies in the ability to
automatically find relational invariants
and complex self-composition derivations.
Moreover, the translation
of programs into Horn clauses,
done by tools
like \reve\footnote{\url{https://formal.kastel.kit.edu/projects/improve/reve/}},
results in
formulas similar to those generated by our
VCGen.
Therefore, like our approach,
relational verification with constraint solving
requires no additional separation hypothesis in presence of aliasing.

\vspace{-2mm}
\paragraph{Certified Verification Condition Generation.}
In a broad sense, this work continues previous
efforts
in formalization and
mechanized proof of program language semantics, analyzers and
compilers, 
such as~\cite{SF,leroy-formal-2008,herms-certification-2013,BeringerA19,JourdanLBLP15,JungKJBBD18,WilsJacobsCertCProg2021,KrebbersLW14,BlazyMP15,ParthasarathyMuellerSummers21}.
%
%
%
%
%
%
%
%
Generation of certificates (in Isabelle)
for the \textsc{Boogie} verifier is presented 
in~\cite{ParthasarathyMuellerSummers21}.
The certified deductive
verification tool WhyCert~\cite{herms-certification-2013}
comes with a similar soundness result
for its verification condition generator.
Its formalization follows an alternative proof approach,
based on co-induction, while our proof relies  on induction.
WhyCert is syntactically closer to the C language and the
\acsl specification language~\cite{ACSL},
while our proof uses a simplified
language, but with a richer aliasing model.
Furthermore, we provide a formalization and a soundness
proof for relational verification, which was not considered
in WhyCert or in~\cite{ParthasarathyMuellerSummers21}.

Our previous work~\cite{BlatterKPLiFM22} 
presented a method for relational property verification based on a naive VCGen.
To the best of our knowledge, the present work
is the first proposal of
\emph{modular} relational property verification based on an \emph{optimized} VCGen
for a representative language with procedure calls and aliases
with a full mechanized formalization and proof of soundness in \coq.

\vspace{-3mm}
\section{Conclusion} 
\vspace{-2mm}
\label{sec:conclusion}
We have presented in this paper
an overview of a method for modular verification of relational properties using 
an optimized verification
condition generator, without relying on code transformations (such as
self-composition) or making additional
separation hypotheses in case of aliasing.
This method has been fully formalized in \coq, and the soundness of recursive
relational verification using a verification condition generator
(itself formally proved correct) for a simple language with procedure calls and aliasing
has been formally established.

This work opens the door for interesting future work.
Currently, for relational properties, product programs~\cite{DBLP:conf/fm/BartheCK11}
or other self-composition optimizations~\cite{ShemerCAV2021}
are the standard approach to deal with complex loop constructions.
We expect that user-provided coupling invariants and
loop properties can avoid having to rely on code transformation methods.
Showing this in our framework is the next step, before the
investigation of termination and co-termination
\cite{DBLP:conf/cade/HawblitzelKLR13},\cite{HiroshCAV2021} for
extending the modularity of relational contracts.

\vspace{-3mm}
\bibliographystyle{splncs04}
\bibliography{biblio}

\clearpage
\appendix

\begin{center}
  {\huge \bf Appendix}
\end{center}
\noindent
This appendix is provided for convenience of the reviewers, not for publication.
\section{Complete Semantics of Language $\langname$}
\label{app-sem}

\subsection{Evaluation of Arithmetic and Boolean Expressions in $\langname$}
\label{app-sem-exp-eval}

We provide
a complete list of rules for evaluation of arithmetic and Boolean expressions 
in $\langname$ in Fig.~\ref{fig:exp-eval-full}.
Evaluation of arithmetic and Boolean expressions in $\mathcal L$ is defined by functions
$\EAexp$ and $\EBexp$.
As mentioned above, the subtraction is lower-bounded by 0.
Operations $*\locval_i$ and $\&\locval_i$ have a
semantics similar to the C language, {i.e.} dereferencing and address-of.
Semantics of Boolean expressions is standard~\cite{DBLP:books/daglib/0070910}.

%
%

\begin{figure}[h]
  \centering
  \begin{minipage}{6cm}
    \begin{align*}
      \EAexpf{\nat}{\mem} & \triangleq \nat\\
      \EAexpf{\locval_i}{\mem} & \triangleq \mem(i)\\
      \EAexpf{*\locval_i}{\mem} & \triangleq \mem(\mem(i))\\
      \EAexpf{\&\locval_i}{\mem} & \triangleq i\\
      \EAexpf{\aexp_1 \opa \aexp_2}{\mem} & \triangleq
                                            \EAexpf{\aexp_1}{\mem} \opa \EAexpf{\aexp_2}{\mem}
    \end{align*}
  \end{minipage}
  \begin{minipage}{4cm}
    \begin{align*}
      \EBexpf{true}{\mem} & \triangleq \True\\
      \EBexpf{false}{\mem} & \triangleq \False\\
      \EBexpf{\aexp_1 \opb \aexp_2}{\mem} & \triangleq
                                            \EAexpf{\aexp_1}{\mem} \opa \EAexpf{\aexp_2}{\mem}\\
      \EBexpf{\bexp_1 \opl \bexp_2}{\mem} & \triangleq
                                            \EBexpf{\bexp_1}{\mem} \opl \EBexpf{\bexp_2}{\mem}\\
      \EBexpf{\neg\bexp}{\mem} & \triangleq  \neg\EBexpf{\bexp}{\mem}
    \end{align*}
  \end{minipage}
  \vspace{-1mm}
  \caption{Evaluation of arithmetic and Boolean expressions in $\langname$.}
  \vspace{-5mm}
  \label{fig:exp-eval-full}
\end{figure}

\subsection{Operational Semantics of Commands in $\langname$ in $\langname$}
\label{app-sem-sper-sem}

We provide
a complete operational semantics of commands in $\langname$ in Fig.~\ref{fig:oper-sem-full}.

\begin{figure}[tb]
  \centering
  \begin{minipage}{3cm}
    \begin{equation*}
      \inferrule*
      {}
      {\eval{\wskip}{\memvar}{\memcom}{\memvar}}
    \end{equation*}
  \end{minipage}
  \begin{minipage}{4cm}
    \begin{equation*}
      \inferrule*
      {\EAexpf{\aexp}{\memvar} = \nat}
      {\eval{\locval_i := \aexp}{\memvar}{\memcom}{\memvar[i/n]}}
    \end{equation*}
  \end{minipage}
  \begin{minipage}{4cm}
    \begin{equation*}
      \inferrule*
      {\EAexpf{\aexp}{\memvar} = \nat}
      {\eval{*\locval_i := \aexp}{\memvar}
        {\memcom}\memvar[\memvar(i)/n]}
    \end{equation*}
  \end{minipage}

  \begin{minipage}{5cm}
    \begin{equation*}
      \inferrule*
      {}
      {\eval{\wassert{P}}{\memvar}
        {\memcom}{\memvar}}
    \end{equation*}
  \end{minipage}
  \begin{minipage}{6cm}
    \begin{equation*}
      \inferrule*
      {\EBexpf{\bexp}{\memvar} = \True\\
        \eval{\com_1}{\memvar_1}{\memcom}\memvar_2}
      {\eval{\wif{\bexp}{\com_1}{\com_2}}{\memvar_1}{\memcom}\memvar_2}
    \end{equation*}
  \end{minipage}

  \begin{minipage}{5cm}
    \begin{equation*}
      \inferrule*
      {\eval{\com_1}{\memvar_1}{\memcom}\memvar_2\\
        \eval{\com_2}{\memvar_2}{\memcom}\memvar_3}
      {\eval{\com_1;\com_2}{\memvar_1}{\memcom}\memvar_3}
    \end{equation*}
  \end{minipage}
  \begin{minipage}{6cm}
    \begin{equation*}
      \inferrule*
      {\EBexpf{\bexp}{\memvar} = \False\\
        \eval{\com_2}{\memvar_1}{\memcom}\memvar_2}
      {\eval{\wif{\bexp}{\com_1}{\com_2}}{\memvar_1}{\memcom}\memvar_2}
    \end{equation*}
  \end{minipage}

  \begin{equation*}
    \inferrule*
    {\EBexpf{\bexp}{\memvar_1} = \True\\
      \eval{\com_1}{\memvar_1}{\memcom}\memvar_2\\
      \eval{\wwhilei{\bexp}{P}{\com}}{\memvar_2}{\memcom}\memvar_3}
    {\eval{\wwhilei{\bexp}{P}{\com}}{\memvar_1}{\memcom}\memvar_3}
  \end{equation*}

  \begin{minipage}{6cm}
    \begin{equation*}
      \inferrule*
      {\EBexpf{\bexp}{\memvar} = \False}
      {\eval{\wwhilei{\bexp}{P}{\com}}{\memvar}{\memcom}\memvar}
    \end{equation*}
  \end{minipage}
  \begin{minipage}{5.5cm}
    \begin{equation*}
      \inferrule*
      {\eval{\body{\loccom}{\memcom}}{\memvar_1}{\memcom}\memvar_2}
      {\eval{\wcall{\loccom}}{\memvar_1}{\memcom}\memvar_2}
    \end{equation*}
  \end{minipage}
  \vspace{-1mm}
  \caption{Operational semantics of commands in $\langname$.}
  \vspace{-5mm}
  \label{fig:oper-sem-full}
\end{figure}



\end{document}